\documentclass[10pt]{article}
\usepackage[all]{xy}
\usepackage[retainorgcmds]{IEEEtrantools}
\usepackage[charter]{mathdesign}
\usepackage[oldstylemath, widermath, slantedGreeks, partialup, oldstylenums]{kpfonts}
\usepackage{amsthm, cancel}
\usepackage{pstricks,graphicx,array,blindtext}

\newcommand{\actson}{\
					 \begin{xy}
						{\ar@{->}@/_{9pt}/(4.9,1.5);(5.5,-.6)} 
				          \end{xy}  
				     \ }
\newcommand{\bb}{\boldsymbol}
\newcommand{\m}{\mathbb}

\newcommand{\whector}[1]{\overset{\rightharpoonup}{#1}}

\let\phi\varphi
\let\theta\vartheta
\renewcommand{\thefootnote}{\arabic{footnote}}

\begin{document}				     
\bibliographystyle{plain}
\numberwithin{equation}{section}

\newcounter{thm}
\newcounter{lemma}
\newcounter{remark}
\numberwithin{remark}{subsection}
\newtheorem{theorem}{Theorem}[section]
\newtheorem{proposition}[theorem]{Proposition}
\newtheorem{lemma}[theorem]{Lemma}
\newtheorem{cor}[theorem]{Corollary}
\newtheorem{conjecture}[theorem]{Conjecture}

\theoremstyle{definition}
\newtheorem{definition}[theorem]{Definition}
\newtheorem{example}[theorem]{Example}
\newtheorem{xca}[theorem]{Exercise}
\newtheorem{problem}[theorem]{Problem}
\newtheorem{remark}[theorem]{Remark}
\newtheorem{properties}{Properties}

\long\def\symbolfootnote[#1]#2{\begingroup\def\thefootnote{\fnsymbol{footnote}}\footnote[#1]{#2}\endgroup}

\title{An Exact Expression for Photon Polarization in Kerr geometry}
\author{Anusar Farooqui\footnote{Department of Mathematics and Statistics, McGill University, 805 Sherbrooke Street West, Montreal QC H3A 0B9, Canada. Email: farooqui@math.mcgill.ca }, Niky Kamran\footnote{Department of Mathematics and Statistics, McGill University, 805 Sherbrooke Street West, Montreal QC H3A 0B9, Canada. Email: nkamran@math.mcgill.ca}, and Prakash Panangaden\footnote{School of Computer Science, McGill University, 3480 University Street, Montreal QC H3A 0E9, Canada. Email: prakash@cs.mcgill.ca}}
\date{\today}

\maketitle
\begin{abstract}

  We analyze the transformation of the polarization of a photon propagating
  along an arbitrary null geodesic in Kerr geometry. The motivation comes
  from the problem of an observer trying to communicate quantum information
  to another observer in Kerr spacetime by transmitting polarized photons.
  It is essential that the observers understand the relationship between
  their frames of reference and also know how the photon's polarization
  transforms as it travels through Kerr spacetime.  Existing methods to
  calculate the rotation of the photon polarization (Faraday rotation)
  depend on choices of coordinate systems, are algebraically complex and
  yield results only in the weak-field limit.

  We give a closed-form expression for a parallel propagated frame along an arbitrary
  null geodesic using Killing-Yano theory, and thereby solve the problem of
  parallel transport of the polarization vector in an intrinsic, 
  geometrically-motivated fashion. The symmetries of Kerr geometry are
  utilized to obtain a remarkably compact expression for the
  geometrically induced phase of the photon's polarization.  We show that
  this phase vanishes on the equatorial plane and the axis of symmetry.

\end{abstract}

\newpage
\tableofcontents{}
\section{Introduction}

In protocols for quantum communication~\cite{Nielsen00} most of the
attention has been focussed on quantum effects such as the problem of
coping with noise in the communication mechanism or preserving
entanglement. It is typically taken for granted that the participants in
the protocol share a frame of reference.  However, a closer analysis by
Bartlett et al.~\cite{Bartlett03,Bartlett07} has revealed the importance of
sharing a frame. They have even quantified the degree to which a partially
shared frame constitutes shared information. The present paper is
motivated by these considerations, however, we do not address the quantum
information-theoretic issues which would involve a study of the evolution
of the quantum state.

Instead, we isolate the \emph{classical geometric aspects} and study them in
Kerr geometry.  Specifically, we study how two participants in a
quantum communication protocol involving transmission of polarized photons
--- henceforth we will call them \emph{observers} --- could share a frame
in Kerr geometry and how the polarization of a linearly polarized photon
would transform as it travels from one observer to the other.  It is
crucial that this transformation reflect what would be \emph{seen by the
  observers}.  Furthermore, the quantity we report should be
\emph{intrinsic to the geometry of the spacetime} and not correspond to
some arbitrarily chosen coordinate system or frame.
 
We have two observers called Alice and Bob.  Alice sends a linearly polarized photon to Bob; she has
chosen the polarization vector to be at some angle in with respect to some
axes which she has chosen in the plane of polarization of the photon, which
is a $2$-plane orthogonal to the direction of propagation of the photon.
Bob receives this photon which has travelled through the Kerr spacetime to
reach him.  In order for Bob to measure the polarization of the photon and
know what angle Alice intended to communicate to him, he needs to know how
their frames correspond and how the photon polarization has been
transformed by the background geometry.  

In Minkowski geometry, the problem is straightforward. Since, the background geometry does not affect the polarization of the photon, one only needs to solve the problem of sharing frames. For a pair of observers who start at the same event with a known relation between their frames, one can Fermi-Walker transport their frames to determine how their frames relate at the point where photon transmission occurs.

At this level of generality the problem is intractable in the Kerr geometry since the Fermi-Walker transport of frames along general timelike curves in Kerr geometry is still an open problem. Note that we are not interested in obtaining reference frames per se. What is required for two observers to exchange quantum information using polarized photons is shared knowledge of basis vectors in which the measurement is performed. 

Thus, we seek an \emph{intrinsic, geometrically defined measurement basis} along specific trajectories. We show how Kerr geometry allows for such a protocol; one which simultaneously solves the problem of sharing frames and minimizes the informational requirement on the observers.

The gravitationally induced rotation of the polarization vector in Kerr
geometry has been investigated in the weak field limit by \cite{Skrotskii57},\cite{Plebanski60},\cite{Godfrey70},\cite{Fayos82}, \cite{Ishihara88},
\cite{Nouri-Zonoz99},\cite{Sereno04}.  Extant methods rely on the existence of the Walker-Penrose conserved quantity to solve the problem of parallel propagating the polarization vector along a null
geodesic~\cite{Chandrasekhar92}.  The estimates are difficult to reconcile because they do not take into account the role of reference frames. However, there is a virtual consensus that the acquired phase is zero in
Schwarzschild geometry.

More recently, Brodutch, Demarie and Terno \cite{Brodutch11} have chosen
observers equipped with an orthonormal frame, located at fixed values of
$(r,\theta,\phi)$ in Boyer-Lindquist coordinates.  They make a physically
motivated choice of basis vectors for the plane of polarization by
\emph{requiring} that the acquired phase be zero in the Schwarzschild
limit.  They find that the acquired phase on the equatorial place is zero
and argue that this is because motion on the equatorial plane in Kerr
geometry is qualitatively similar to Schwarzschild spacetime.  Though their
results only hold in the weak field limit, their paper represents a
significant advance in the field.

We take a geometric approach.  First, our choice of observers is
motivated by the intrinsic geometry of Kerr spacetime.  We choose a class
of observers whose 4-velocities are symmetric linear combinations of the
principal null directions of the Weyl tensor.  We show how this class of
observers is uniquely suited to analyze the behaviour of test particles
near the horizon.  Second, we endow these observers with a symmetric frame
by exploiting the existence of the involutive isometry obtained by
simultaneous time- and rotation-reversal of the Kerr black hole.  This
greatly simplifies our expressions.  Third, we use the existence of the
principal null directions and other special features of Kerr geometry
to fix the definition of the plane of polarization and of the measurement basis.  This measurement protocol is allowed by the
specific symmetry structure of Kerr geometry.  It is simply unavailable in
Minkowski spacetime where no direction is similarly privileged.  Fourth, we
use Killing-Yano theory to construct a parallel propagated frame along the
null geodesic (a problem first solved by Marck \cite{Marck83a} in a related but slightly different form), thereby reducing the transport problem to one of raising and
lowering frame indices.  This allows us to obtain a remarkably compact
exact expression for Faraday rotation everywhere in the zone of outer communication in Kerr spacetime. 

We proceed as follows.  Section 2 lays out the geometry and symmetries of
the Kerr solution, as well as describes the null geodesic equations.  The
construction of the parallel propagated frame is given in section 3.  In
section 4, we set out our choice of observers and the measurement protocol.
We then prove that there is no Faraday rotation for photons confined to the
equatorial plane and the axis of symmetry.  We show how this immediately
implies the vanishing of the acquired phase in Schwarzschild spacetime as
well.  Section 5 gives the derivation of the closed form expression for the
Faraday rotation.  In section 6, we discuss the plots of a few null geodesics (provided in Appendix B) and their associated Faraday rotation. We conclude with some remarks about the physical
significance of the results and some possible avenues for future work.

\section{Kerr geometry}\label{Kerr}

Our goal in this section is to recall some of the salient geometric
properties of the Kerr metric that will be used to calculate the Faraday
rotation undergone by the polarization vector of a photon.  (Throughout this
paper, a photon will be thought of as a classical zero rest mass particle
moving along an affinely parametrized null geodesic.) We shall see that the
remarkable symmetry and separability properties of Kerr geometry make
it possible to obtain an exact expression for the Faraday rotation, which
we will derive in Section~\ref{ExactFaraday} and will subsequently
interpret geometrically.

We begin by recalling that the Kerr metric is a two-parameter family of
solutions of the Einstein vacuum equations defined on the manifold
$M\equiv \mathbb{R}^{2}\times S^{2}$ and describing the outer geometry
of a rotating black hole in equilibrium.  In Boyer-Lindquist coordinates
$(x^{i})=(t, r, \mathit\theta, \phi)$ with $-\infty<t<+\infty$,
$r_{+}<r<+\infty$, $0\le\theta\le\pi$, $0\le\phi<2\pi$, the Kerr metric
takes the form
\begin{equation}
\label{}
\nonumber
 ds^{2}= \frac{\Delta}{\Sigma}\left(dt-a\sin^{2}\theta d\phi\right)^{2} - \frac\Sigma\Delta dr^{2} -\Sigma d\theta^{2} -\frac{\sin^{2}\theta}\Sigma\left(adt-\left(r^{2}+a^{2}\right)d\phi\right)^{2},
\end{equation}
with
\begin{equation}
\Sigma(r, \theta) = r^{2}+a^{2}\cos^{2}\theta, \quad
\Delta(r) =  r^{2}-2Mr+a^{2}.
\end{equation}
The parameters $M>0$ and $a\ge0$ labeling the solutions within the Kerr
family correspond respectively to the mass and angular momentum per unit
mass of the black hole, as measured from infinity.  We shall restrict our
attention to the \emph{non-extreme case} $M>a\ge0$, in which case the
function $\Delta(r)$ has two distinct zeros,  
\begin{equation}
\label{ }
 r_{\pm}=M\pm\sqrt{M^{2}-a^{2}},
\end{equation}
with $r_{+}$ corresponding to the lower limit of the range of the
Boyer-Lindquist coordinate $r$.  It is well-known that the Kerr metric can
be analytically continued across the hypersurfaces $r=r_{+}$ and $r=r_{-}$
in such a way that these become null hypersurfaces in the extended
manifold, corresponding respectively to the event and Cauchy horizons of
the Kerr black hole geometry.  We shall however be interested in the
region $r>r_{+}$, which describes the space-time geometry outside the event
horizon of a Kerr black hole with parameters $M$ and $a$.  

The Weyl conformal curvature tensor of the Kerr solution is of Petrov type
D, meaning that it admits a pair of repeated principal null
directions, each of which is defined up to multiplication by a non-zero
scalar function.  These repeated principal null directions give rise to null
congruences which are geodesic and shear-free as a consequence of the
Goldberg-Sachs Theorem.  We choose the scale factors in such a way that the
principal null directions are given by  
\begin{equation}\label{lsymm}
\bb\ell=\ell^{i}\frac{\partial}{\partial
  x^{i}}=\frac{1}{\sqrt{2\Sigma\Delta}}\left(\left(r^{2}+a^{2}\right)\frac{\partial}{\partial
    t}+\sqrt{\Delta}\frac{\partial}{\partial r}+a\frac{\partial}{\partial
    \varphi}\right), 
\end{equation}
and
\begin{equation}\label{nsymm}
\bb n=n^{i}\frac{\partial}{\partial
  x^{i}}=\frac{1}{\sqrt{2\Sigma\Delta}}\left(\left(r^{2}+a^{2}\right)\frac{\partial}{\partial
    t}-\sqrt{\Delta}\frac{\partial}{\partial r}+a\frac{\partial}{\partial
    \varphi}\right).  
\end{equation}
\\ \\
The choice of scale factors leading to (\ref{lsymm}) and (\ref{nsymm}) will be characterized geometrically through an involutive isometry admitted by the Kerr metric (see (\ref{invol})). In particular, the vector fields (\ref{lsymm}) and $(\ref{nsymm})$, which form part of the \emph{symmetric null frame} constructed by Debever et al. \cite{Debever79}, will play an important in the geometrical characterization of the class of observers that we shall consider in our calculation of the Faraday rotation.  

The Kerr metric enjoys remarkable symmetry properties which we will exploit
systematically in our calculation of the Faraday rotation and which we now
summarize.  

First of all, the Kerr metric admits a two-parameter Abelian isometry group
that acts orthogonally transitively on time-like orbits, meaning that the
orbits of the group action are time-like 2-surfaces with the property that
the distribution of 2-planes orthogonal to the orbits is integrable.  The
orthogonal transitivity is manifest in the Boyer-Lindquist coordinates
since the metric does not admit cross terms mixing the differentials
$dr,\,d\theta$ with the differentials $dt,\,d\phi$.  In Boyer-Lindquist
coordinates, the action of the continuous part of the isometry group is
generated by the flows of the pair of commuting Killing vector fields
$\partial_{t}$ and $\partial_{\phi}$, and thus given by 
\begin{equation}
\label{transl}
(t,r,\theta, \phi)\mapsto(t+c_{1}, r, \theta, \phi+c_{2}),
\end{equation}
where $c_{1},\,c_{2}$ are arbitrary real constants.  Furthermore, the
isometry group of the Kerr metric admits a discrete subgroup isomorphic to
$\mathbb{Z}_{2}$, whose action is not of the form (\ref{transl}).  More
precisely, we say following Carter's terminology \cite{Carter69} that the isometry group is
\emph{invertible}, meaning that at every $x\in M$, there exists a
(1,1)-tensor $L_{x}\in\text{End}(T_{x}M)$, which acts as an involutive
isometry of $(T_{x}M,g_{x})$ and is such that if $\mathcal{O}_{x}$
denotes the orbit of the isometry group through $x$, then 
\begin{equation}
\label{ }
L_{x}\vert_{(T_{x} \mathcal{O}_{x})^{\perp}}=\text{id}_{(T_{x}\mathcal{O}_{x})^{\perp}},
\end{equation} 
and for all $X_{x}\in(T_{x}\mathcal{O}_{x})^{\perp}$
\begin{equation}
\label{ }
L_{x}(X_{x})=-X_{x}.
\end{equation}
We remark that from a result of Carter \cite{Carter69}, it is known that if an isometry
group acts orthogonally transitively on non-null orbits then the action is
necessarily invertible\footnote{This result is not true if the orbits of
  the isometry group are null}.  In Boyer-Lindquist coordinates, the
involution is given by 
\begin{equation}
\label{involution}
L_{x}=f_{*}\vert_{x},
\end{equation}
where $f$ is the isometry given by 
\begin{equation}
\label{ }
(t,r,\theta, \phi)\mapsto(-t, r, \theta, -\phi).
\end{equation}
We will commit an abuse of notation and denote both $L_{x}=f_{*}\vert_{x}$
and the dual map $f^{*}\vert_{x}$ by $L$.  The involution $L$ will play a
key role in defining invariantly the class of observers and frames for
which the Faraday rotation will be computed.  \\
We shall work in a Newman-Penrose null coframe
\begin{equation}
\label{ }
\theta^{1}=n_{i}dx^{i}, \ \theta^{2}=\ell_{i}dx^{i}, \ \theta^{3}=-\bar m_{i}dx^{i}, \ \theta^{4}=-m_{i}dx^{i},
\end{equation}
in which the Kerr metric takes the form 
\begin{equation}
\label{ }
ds^{2}=2(\theta^{1}\theta^{2}-\theta^{3}\theta^{4}).
\end{equation}
Following the construction of Debever et al. \cite{Debever79}, this coframe is chosen such that 
 \begin{equation}
\label{invol}
L\theta^{1}=-\theta^{2}, \ L\theta^{2}=-\theta^{1}, \ L\theta^{3}=-\theta^{4}, \ L\theta^{4}=-\theta^{3}.
\end{equation}
Following \cite{Debever79}, we refer to this frame as the \emph{symmetric coframe}.  
Note that this last requirement eliminates the scaling freedom we would
have otherwise had in defining a null coframe adapted to the principal null directions of the Weyl tensor. The corresponding orthonormal symmetric coframe $(\bb\omega^{0},\bb\omega^{1},\bb\omega^{2},\bb\omega^{3})$ is then defined by 
\begin{equation}
\label{ }
\bb\omega^{0}=\frac{1}{\sqrt{2}}(\theta^{1}+\theta^{2}), \bb\omega^{1}=\frac{1}{\sqrt{2}}(\theta^{2}-\theta^{1}), \bb\omega^{2}=-\frac{1}{\sqrt{2}}(\theta^{3}+\theta^{4}), \bb\omega^{3}=\frac{1}{\sqrt{2}}(\theta^{3}-\theta^{4}),
\end{equation} 
and given in Boyer-Lindquist coordinates by
\begin{eqnarray}
\label{carter11}
\bb\omega^{0} & = & \sqrt{\frac{\Delta}{\Sigma}}\left(dt-a\sin^{2}\theta d\varphi\right),\\
\bb\omega^{1} & = &  \sqrt{\frac{\Sigma}{\Delta}}dr,\\
\bb\omega^{2} & = & \frac{\sin\theta}{\sqrt\Sigma}\left(a dt-(r^{2}+a^{2})d\varphi\right),\\
\label{carter12}
\bb\omega^{3} & = &  \sqrt\Sigma\ d\theta.
\end{eqnarray}
Throughout this paper, we shall reserve lower case Latin indices $a,b,c,\dots$ to denote components with respect to the orthonormal symmetric coframe (\ref{carter11})-(\ref{carter12}) and the orthonormal frame dual to it. We shall denote the flat spacetime metric used to raise and lower these orthonormal frame indices by $\eta$, where
\begin{equation}
\label{ }
\eta_{ab}=\eta^{ab}=\left[\begin{array}{cccc}1 &  &  &  \\ & -1 &  &  \\ &  & -1 &  \\ &  &  & -1\end{array}\right].
\end{equation}

It is well known that in addition to its two-parameter Abelian group of
isometries, the Kerr metric posesses further symmetries whose presence is
closely tied to the fact that all the known massless and massive wave
equations are separable in Boyer-Lindquist coordinates and, where
applicable, the symmetric frame.  The geometic object that generates all
these additional symmetries is a rank two \emph{Killing-Yano tensor}, that
is, a $(0,2)$ skew-symmetric tensor $(f_{ij})$ satisfying the Killing-Yano
equation 
\begin{equation}
\label{KYE}
\nabla_{i}f_{jk}+\nabla_{j}f_{ik}=0.
\end{equation}
In Boyer-Lindquist coordinates and in the symmetric orthonormal coframe, any rank 2 Killing-Yano tensor is a constant multiple of
\begin{equation}
\label{KYT}
\bb f:=\frac12 f_{ij}dx^{i}\wedge dx^{j}=-a\cos\theta\bb\omega^{0}\wedge\bb\omega^{1}+r\bb\omega^{2}\wedge\bb\omega^{3}
\end{equation}
The role played by this Killing-Yano tensor in the separability properties of the Kerr metric stems from the fact that it appears as a ``square root'' of the quadratic first integral discovered by Carter in his proof of the separability in Kerr geometry of the Hamilton-Jacobi equation for geodesics and the Klein-Gordon equation for massive scalar fields \cite{Carter68h}.  More precisely, the symmetric (0,2)-tensor $(K_{ij})$ defined by
\begin{equation}
K_{ij}=f_{ik}f^{k}_{{\phantom j}j},
\end{equation}
satisfies the Killing equation
\begin{equation}
\nabla_{i}K_{jk}+\nabla_{j}K_{ki}+\nabla_{k}K_{ij}=0,
\end{equation}
and therefore gives rise to a quadratic first integral for the geodesic flow in Kerr geometry first discovered by Carter \cite{Carter68g}
\begin{equation}
\label{fourthintegral}
\kappa=K^{ij}p_{i}p_{j}.
\end{equation}
This quadratic first integral exists in addition to the two linear first integrals arising from the presence of the two commuting Killing vector
fields $\partial_{t}$ and $\partial_{\phi}$ and therefore reduces the
integration of the geodesic flow to quadratures.  For the purposes of
calculating the Faraday rotation of a photon, we shall be interested in
affinely parametrized \emph{null} geodesics, for which the equations can be
written in first-order form as  
\begin{eqnarray}
\label{geodesiceqns1}
\dot r & = & \pm\frac{\sqrt{R}}{\Sigma}, \\
\dot\theta & = & \pm\frac{\sqrt{\Theta}}{\Sigma}, \\
\Sigma\Delta\dot t & = & E\left(\left(r^{2}+a^{2}\right)^{2}-\Delta a^{2}\sin^{2}\theta\right)-2Mra\Phi, \\
\label{geodesiceqns2}
\Sigma\Delta\dot\phi & = & 2MraE+\left(\Sigma-2Mr\right)\Phi/\sin^{2}\theta.
\end{eqnarray}
In these equations, the dot denotes the derivative with respect to an affine parameter $s$, the constant $E$ is the conserved momentum $p_{t}$ corresponding to the energy of zero rest-mass particle moving along the null geodesic, $\Phi$ is the conserved angular momentum $-p_{\phi}$ along the axis of symmetry of the Kerr black hole, $\kappa$ is Carter's fourth integral of motion given by (\ref{fourthintegral}), and
\begin{eqnarray}
\label{R}
R(r)&:=&\m{P}^{2}-\Delta\kappa, \\
\label{Theta}
\Theta(\theta)  & := & \kappa - \m{D}^{2},
\end{eqnarray}
where  
\begin{eqnarray}
\m{P}(r) &:=& E(r^{2}+a^{2})-a\Phi, \\
\m{D}(\theta)  & := & a\sin\theta E - \Phi/\sin\theta.
\end{eqnarray}

\section{Parallel-propagated frame along null geodesics}
By definition, the \emph{polarization 4-vector} $\bb\digamma$ of a photon is
a vector field along an affinely parametrized null geodesic $\gamma$ with tangent vector $\bb K$ that is both parallel propagated along $\gamma$ and orthogonal to $\bb K$, that is,
\begin{eqnarray}
\label{TransportLaw}
K^{i}\nabla_{i}\digamma^{j}=K^{a}\nabla_{a}\digamma^{b} & = & 0, \\
\label{OrthogonalityCondition}
K^{i}\digamma^{j}g_{ij}=K^{a}\digamma^{b}\eta_{ab} & = & 0. 
\end{eqnarray}
In order to solve this transport equations (\ref{TransportLaw}) and (\ref{OrthogonalityCondition}), we construct a frame that is
parallel propagated along an arbitrary null geodesic in Kerr geometry. 
We shall see that, just as in Marck's original construction
\cite{Marck83a}, the Killing-Yano tensor (\ref{KYE}) will play a key
role. 

We first recall that the two commuting Killing vectors admitted by the Kerr metric can be recovered from the Killing-Yano tensor (\ref{KYE}) using the Hodge duality operator.  Indeed, it follows from the defining equation (\ref{KYE}) for Killing-Yano tensors that the vector fields $\bb\xi$ and $\bb\zeta$ defined by 
\begin{equation}
\label{primary}
\xi^{i}:=\frac{1}{3}\nabla_{j}h^{ji}, \quad \zeta^{i}:=-K^{i}_{\phantom{a}j}\xi^{j},
\end{equation}
where $(h_{ij})$ denotes the Hodge dual of $(f_{ij})$, are Killing vector fields.  Explicitly, with the Killing-Yano tensor $(f_{ij})$ given by (\ref{KYT}), the Hodge dual $\bb h$ is given by 
\begin{equation}
\label{ }
\bb h= \frac12 h_{ij}dx^{i}\wedge dx^{j}=r\bb\omega^{o}\wedge\bb\omega^{1}+a\cos \theta\bb\omega^{2}\wedge \bb\omega^{3},
\end{equation}
and we have
\begin{equation}
\bb\xi=\partial_{t},\quad \bb\zeta=\partial_{\phi}.
\end{equation}
A parallel propagated frame along the null geodesics of the Kerr metric is now constucted as follows. We follow the construction of Kubiznak et al. \cite{Kubiznak09}. The relevant result is:
\begin{lemma}
\label{KYtheorem} Let $\gamma$ be an affinely parametrized null geodesic with tangent vector $\bb K$.  Let $\bb X$ be a vector field that is both parallel propagated along  $\gamma$ 
\begin{equation}
\label{ }
K^{i}\nabla_{i}X^{j}=0,
\end{equation}
and orthogonal to $\bb K$, 
\begin{equation}
\label{ }
 g_{ij}K^{i}X^{j}=0.
\end{equation}
Then, the vector field $\bb Y$ defined along $\gamma$ by
\begin{equation}
\label{KYvf}
Y^{i}=X^{j}h_{j}^{\phantom{i}i}+\beta_{\bb X} K^{i},
\end{equation}
where
\begin{equation}
\label{ }
\frac{d}{ds}\beta_{\bb X}=g_{kl}X^{k}\xi^{l},
\end{equation}
and $\frac{d}{ds}$ denotes differentiation with respect to an affine parameter $s$ along $\gamma$, and $\bb\xi$ is as defined by (\ref{primary}), is parallel propagated along $\gamma$.   
\end{lemma}
We now consider an affinely parametrized arbitrary null geodesic $\gamma$ in the Kerr metric and construct a parallel propagated frame along $\gamma$ by repeated application of Lemma \ref{KYtheorem}. From now on, we will work exclusively in Carter's symmetric frame, defined as the orthonormal frame dual to the symmetric orthonormal coframe given by (\ref{carter11})-(\ref{carter12}). Vector fields will thus be identified with their components in the symmetric frame and will be represented as four-component row vectors.

Given an affinely parametrized null geodesic $\gamma$, it follows from (\ref{geodesiceqns1})-(\ref{geodesiceqns2}) that the tangent vector $\bb K=\dot\gamma$ is given by
\begin{equation}
\label{ }
\bb K = \frac1{\surd\Sigma}\left(\frac{\m P}{\surd\Delta}, \frac{\surd R}{\surd\Delta}, \m D, \surd\Theta\right).
\end{equation}
Likewise, the Killing vector field $\bb\xi=\partial_{t}$ is given by
\begin{equation}
\label{ }
\bb\xi = \frac1{\surd\Sigma}\left(\surd\Delta, 0, a\sin\theta, 0\right).
\end{equation}
Since $\bb K$ is both parallel propagated along $\gamma$ and null, we may apply Lemma \ref{KYtheorem} to obtain a vector field $\bb Y$ that is parallel propagated along $\gamma$.  We have
\begin{equation}
\label{ }
\frac{d}{ds}\beta_{\bb K}=\eta_{ab}K^{a}\xi^{b}=E,
\end{equation}
so that $\beta_{\bb K}=Es$
where $s$ is the affine parameter of the null geodesic.  We then immediately obtain using (\ref{KYvf}) that the vector field $\bb Y$ defined in the symmetric frame by 
\begin{equation}
\label{ }
\bb Y =\frac{1}{\sqrt{\kappa\Sigma}}\left(\frac{Es\m P-r\surd R}{\surd\Delta}, \frac{Es\surd R-r\m P}{\surd\Delta},Es\m D + a\cos\theta\surd\Theta, Es\surd\Theta-a\cos\theta \m D\right),
\end{equation}
is parallel propagated along $\gamma$. We now apply Lemma (\ref{KYtheorem}) to the vector field $Y$ and obtain an additional vector field $\bb X$ that is parallel propagated along $\gamma$.  We have 
\begin{equation}
\label{ }
\frac{d}{ds}\beta_{\bb Y}=\eta_{ab}Y^{a}\xi^{b}=\frac{E^{2}s-r\dot r-a^{2}\cos\theta\sin\theta\dot\theta}{\surd\kappa},
\end{equation}
where the dot denotes differentiation with respect to the affine parameter $s$, whence
\begin{equation}
\beta_{\bb Y}=\frac{{E^{2}s^{2}-r^{2}+a^{2}\cos^{2}\theta}}{2\surd\kappa}.
\end{equation}
We conclude then that 
\begin{align}
\label{ }
\nonumber
\bb X=&\frac1{\kappa\surd\Sigma}\Big(\frac{\m P\beta_{+}-rEs\surd R}{\surd\Delta}, \frac{\surd R\beta_{+}-rEs\m P}{\surd\Delta},\\
&\phantom{\frac1{\kappa\surd\Sigma}\Big(}\m D\beta_{-}+a\cos\theta Es\surd\Theta, \surd\Theta\beta_{-}-a\cos\theta Es\m D\Big),
\end{align}
where 
\begin{equation}
\label{ }
2\beta_{\pm}:=E^{2}s^{2}\pm\Sigma,
\end{equation}
is parallel propagated along $\gamma$.

Note that $\eta_{ab}X^{a}K^{b}=1$ so that $\bb X$ and $\bb K$ are not orthogonal. Thus, we cannot apply Lemma \ref{KYtheorem} to construct a fourth vector field that is parallel propagated along $\gamma$.  However, we can use the Killing-Yano tensor $\bb f$ directly to obtain another vector that is parallel propagated along the affinely parametrized null geodesic $\gamma$ with tangent vector $\bb K$.  Indeed, it follows immediately from the Killing-Yano equation (\ref{KYT}) that the vector field $\bb Z$ defined by 
\begin{equation}
\label{KYP}
Z^{a}=f^{a}_{\phantom{a} b}K^{b},
\end{equation}
is parallel propagated along $\gamma$.  We are of course free to scale $\bb Z$ by any constant, and we will choose this constant to be equal to $1\over {\sqrt{\kappa}}$ so as to simplify the orthogonality relations between the vector fields comprising the parallel propagated frame.  Applying (\ref{KYP}) and scaling $\bb Z$ as above, we obtain
\begin{equation}
\label{ }
\bb Z=\frac{1}{\sqrt{\kappa\Sigma}}\bigg(\frac{a\cos\theta\surd R}{\surd{\Delta}}, \frac{a\cos\theta\m P}{\surd{\Delta}}, r\surd\Theta,-r\m D\bigg).
\end{equation}
We thus have a frame $\{\bb K, \bb X, \bb Y, \bb Z\}$ that is parallel propagated along the affinely parametrized null geodesic $\gamma$ with tangent vector $\bb K$.  The matrix of scalar products for the elements of this frame is given by
\begin{equation}
\label{prodmat}
\left[\begin{array}{cccc} & 1 &  &  \\ 1 &  &  &  \\ &  & -1 &  \\ &  &  &- 1\end{array}\right].
\end{equation}
The polarization vector $\bb\digamma$ is orthogonal to $\bb K$. Since $\bb K$ is null and parallel propagated along itself, $\bb\digamma$ is only determined \emph{modulo} $\bb K$.  That is, if $\bb\digamma$ satisfies (\ref{TransportLaw}) and (\ref{OrthogonalityCondition}), then so does 
\begin{equation}
\label{Pfreedom}
\bb\digamma'=\bb\digamma+c\bb K,
\end{equation}  
where $c$ is a real constant. 
\begin{definition}[2-Plane of Polarization along $\gamma$]
We choose initial conditions such that $\bb\digamma_{\gamma(0)}\in$ span$\left\{\bb Y_{\gamma(0)},\bb Z_{\gamma(0)}\right\}$. Then, $\bb\digamma_{\gamma(s)}\in\text{span}\{\bb Y_{\gamma(s)}, \bb Z_{\gamma(s)}\}$ for all $s$, since $\bb\digamma$ has constant components in $\{\bb K, \bb X, \bb Y, \bb Z\}$. This defines the 2-plane of polarization 
\begin{equation}
\mathcal P_{\gamma(s)}:=\text{span}\left\{\bb Y_{\gamma(s)}, \bb Z_{\gamma(s)}\right\}\subset \langle \bb K\rangle^{\perp},
\end{equation}
 at each event $\gamma(s)\in M$. 
\end{definition}
In order to simplify the computation we define an orthonormal frame that is parallel propagated along the null geodesic $(\bb L_{(0)},\bb  L_{(1)},\bb L_{(2)},\bb L_{(3)})$ by a constant coefficient transformation of $\{\bb K, \bb X, \bb Y, \bb Z\}$ as follows.
\begin{equation}
\label{ }
\bb L_{(0)}:=\frac{1}{\sqrt2}(\bb K+\bb X), \,\bb  L_{(1)}:=\frac{1}{\sqrt2}(\bb K-\bb X), \ \bb  L_{(2)}:=\bb Y,\bb  L_{(3)}:=\bb Z.
\end{equation}
We note here the explicit expressions for the elements of frame $\bb L_{(a)}$ with respect to the symmetric frame.
\begin{eqnarray}
\nonumber
\bb L_{(0)}&=&\frac1{\kappa\sqrt{2\Sigma}}\bigg[\frac{\m P(\kappa+\beta_{+})-Es\surd R}{\surd\Delta}, \frac{\surd R(\kappa+\beta_{+})-Es\m P}{\surd\Delta},\\ 
\nonumber
&&\m D(\kappa+\beta_{-})+a\cos\theta Es\surd\Theta, \surd\Theta(\kappa+\beta_{-})-a\cos\theta Es\m D\bigg], \\
\nonumber
\bb L_{(1)}&=&\frac1{\kappa\sqrt{2\Sigma}}\bigg[\frac{\m P(\kappa-\beta_{+})+Es\surd R}{\surd\Delta},\frac{\surd R(\kappa-\beta_{+})+Es\m P}{\surd\Delta},\\ 
\nonumber
&&\m D(\kappa-\beta_{-})-a\cos\theta Es\surd\Theta,\surd\Theta(\kappa-\beta_{-})+a\cos\theta Es\m D\bigg], \\
\nonumber
\bb L_{(2)}&=&\frac1{\sqrt{\kappa\Sigma}}\bigg[\frac{Es\m P-r\surd R}{\surd\Delta},\frac{Es\surd R-r\m P}{\surd\Delta},Es\m D + a\cos\theta\surd\Theta,Es\surd\Theta-a\cos\theta \m D \bigg], \\
\label{ppframe}
\bb L_{(3)}&=&\frac1{\sqrt{\kappa\Sigma}}\bigg[\frac{a\cos\theta\surd R}{\surd{\Delta}},\frac{a\cos\theta\m P}{\surd{\Delta}},r\surd\Theta,-r\m D\bigg].
\end{eqnarray}

\section{Defining and measuring Faraday rotation}
In order to define the Faraday rotation we need to pin down the class of
observers who are involved in the communication protocol and specify the
frames with respect to which they are measuring the polarization.  We have
already seen that Carter's symmetric frame frame is closely tied to intrinsic
geometric properties of the Kerr metric.  This makes it an ideal candidate for formulating the
Faraday rotation in a geometrically meaningful fashion.

In the definition of Carter's null frame, which is dual
to the co-frame defined in (\ref{carter11})-(\ref{carter12}), the arbitrary scaling of the
vectors $\bb\ell$ and $\bb n$ has been fixed by the action of the involution.
Thus, one has a natural time-like vector field $\bb U$, namely  
\begin{equation}
\label{Observers}
\bb U:=\frac{1}{\sqrt{2}}(\bb\ell+\bb n)=\frac{1}{\sqrt{\Sigma\Delta}}\left(\left(r^{2}+a^{2}\right)
\frac{\partial}{\partial t}+a\frac{\partial}{\partial \varphi}\right), 
\end{equation}
where $\bb\ell$ and $\bb n$ are given by (\ref{lsymm}) and (\ref{nsymm}).  This
identifies a family of observers whose 4-velocities are a symmetric linear
combination of the principal null directions $\bb\ell$ and $\bb n$.  We call them
\emph{Carter observers}.  We choose to work with Carter observers because
it follows from the discussion of Section~\ref{Kerr} that they are defined
geometrically in terms of the principal null directions of the Weyl tensor
and the involution $L$.  They exist everywhere outside the event horizon
\emph{including} the region between the event horizon and the stationary limit surface
$r=r_{s}$ where the Killing vector field $\bb\xi=\partial_{t}$ becomes null.
Their coordinate angular velocity is $a/(r^{2}+a^{2})$, which is exactly
the coordinate angular velocity of the event horizon with $r=r_{+}$.
Therefore, this class of observers is uniquely suited to analyze the
behaviour of test particles near the horizon.

We choose the observers' frames to be duals of the symmetric coframe defined in equations (\ref{carter11})-(\ref{carter12}).  We shall see that since the symmetric frame is so well adapted to the geometry of Kerr spacetime, this choice will greatly simplify the computation and allow us to obtain a compact, closed form expression for the geometrically induced Faraday rotation of the polarization vector.

We decompose the tangent space $T_{x}M$ at any event $x$ along the
worldline of the observer with 4-velocity $\bb U$ into an orthogonal direct sum of spacelike and timelike
vector spaces in accordance with the observer's decomposition of spacetime
by projecting vectors onto the observer's frame at event $x$.  That is,  
\begin{equation}
\label{ }
T_{x}M=\langle\bb U\rangle\oplus\Sigma_{3},
\end{equation}
where $\Sigma_{3}:=\langle\bb U\rangle^{\perp}$.
In what follows, we shall supress the label $x$ for the event with the understanding
that this 1+3 decomposition is only valid at a given event.  
The observer's frame defines a projection map $\pi:T_{x}M\longrightarrow\Sigma_{3}$,
\begin{equation}
\label{proj}
\pi(\bb X)=:\whector{\bb X}=\left[\begin{array}{c}X^{1} \\ X^{2}\\X^{3}\end{array}\right].
\end{equation} 
Given the direction 3-vector of the photon $\pi(\bb K)=\whector{\bb K}$, consider the 2-plane $\langle\whector{\bb K}\rangle^{\perp}\subset\Sigma_{3}$ passing through the origin and orthogonal to it. Given a pair of orthonormal
basis vectors for this 2-plane, $\{\whector{\bb b_{1}}$, $\whector{\bb  b_{2}}\}$, we can write a general polarization 3-vector as 
\begin{equation}
\label{ }
\whector{\bb \digamma}=c_{1}\whector{\bb  b_{1}}+c_{2}\whector{\bb  b_{2}},
\end{equation}
for real constants $c_{1}$ and $c_{2}$. It is crucial this choice of basis vectors not be made arbitrarily.  We
choose basis vectors on intrinsic geometric criteria, which are independent
of coordinate descriptions.  First, we project the principal null direction
$\bb\ell$ onto $\Sigma_{3}$ and obtain the corresponding 3-vector $\pi(\bb\ell):=\whector{\bb\ell}$ using the prescription (\ref{proj}).  Then, we set the basis vectors in the plane of polarization to be 
\begin{eqnarray}
\label{PNDgauge1}
\whector{\bb  b_{1}}&:=&\frac{\whector{\bb\ell}\times\whector{\bb K}}{\Vert\whector{\bb \ell}\times\whector{\bb K}\Vert}, \\
\label{PNDgauge2}
\whector{\bb  b_{2}}&:=&\frac{\whector{\bb K}\times\whector{\bb b_{1}}}{\Vert\whector{\bb K}
\times\whector{\bb b_{1}}\Vert}.  
\end{eqnarray}
We are finally ready to spell out the communication protocol.  Let Alice
and Bob be two Carter observers in the Kerr exterior.  In order to
communicate with Bob, Alice sends a polarized photon along a null geodesic $\gamma(s)$
that intersects with Bob's worldline.  Alice polarizes the photon in the
basis (\ref{PNDgauge1})-(\ref{PNDgauge2}) at the event $x_{A}$, which we denote here by
$\whector\digamma_{\text{in}}$.  When Bob sees the photon he also measures
its polarization by projecting it onto the basis
(\ref{PNDgauge1})-(\ref{PNDgauge2}) at the event $x_{B}$ to obtain $\whector\digamma_{\text{out}}$.
Note that since these bases are defined intrinsically they can agree in
advance on the choice of these bases and set them up locally \emph{without further communication} once they have embarked on their orbits.

The geometrically induced Faraday rotation of the polarization vector of a
photon as it transverses the Kerr exterior from Alice and Bob is then given
by the angle $\chi$ such that 
\begin{equation}
\label{Faraday}
\whector\digamma_{\textup{out}}:=\left[\begin{array}{cc}\cos\chi & -\sin\chi \\\sin\chi & \cos\chi\end{array}\right]\whector\digamma_{\textup{in}}.
\end{equation}

\begin{remark}
The reference frame on which the measurement of Faraday rotation depends so critically is \emph{not} the orthonormal symmetric frame. Rather, it is the \emph{measurement basis} $\left\{\whector{\bb b_{1}}, \whector{\bb b_{2}}\right\}$ which plays the role of the reference frame. 
\end{remark}
\begin{remark}
Using the intrinsic geometry of Kerr to solve the problem of choosing a set
of basis vectors for the plane of polarization as specified in this section
simultaneously solves the problem of sharing frames and minimizes the
informational requirement on the observers.  Note that such a strategy is
simply \emph{unavailable} in Minkowski spacetime where no direction is
similarly privileged: there is too much symmetry.\footnote{Of course, in
  Minkowski space one can use alternative protocols.}  In our case, the type D character of Kerr geometry provides \emph{just enough symmetry} to
allow for the possibility of the present protocol with its minimal
communication requirements.  
\end{remark}
We are now in a position to prove
the following proposition.
\begin{proposition}\label{lemma} Consider observers confined to the equatorial plane $Eq:=\{-\infty<t<+\infty,r_{+}>r>+\infty, \theta=\pi/2, 0\le\phi<2\pi\}$. There is no Faraday rotation for photons confined to the equatorial plane of Kerr geometry.  
\end{proposition}
\begin{proof}
Consider the vector field $\bb V:=-\frac1{r}\partial_{\theta}$. This is unit norm vector field which is orthogonal to the equatorial plane when restricted to it. We will by an abuse of notation use $\bb V$ to denoted $\bb V\vert_{Eq}$. An easy calculation shows that $\pi(\bb V)=\whector{\bb b_{1}}$. Using the expressions for the Christoffel symbols given in Appendix A, we obtain
\begin{equation}
\label{ }
\nabla_{\bb K}\bb V=0.
\end{equation}
By (\ref{TransportLaw}) and (\ref{OrthogonalityCondition}), it follows that $\digamma^{a} V^{b}\eta_{ab}=0$, which together with $\pi(\bb V)=\whector{\bb b_{1}}$ implies
\begin{equation}
\label{ }
\whector{\bb\digamma}\cdot\whector{\bb b_{1}}=\text{constant}.
\end{equation}
\end{proof}

\begin{cor}There is no Faraday rotation in the Schwarzschild geometry.
\begin{proof}
  Since the Schwarzschild geometry is spherically symmetric, geodesics are
  confined to planes through the origin~\cite{Chandrasekhar92}.  Therefore,
  the exact same argument as we used for the equatorial plane in Kerr can
  be used here.  Any plane through the original can be viewed as the
  equatorial plane of a degenerate Kerr solution with $a=0$.
\end{proof}
\end{cor}
\begin{proposition}\label{lemma2} There is no Faraday rotation for orbits
  confined to the axis of symmetry of Kerr geometry.   
\begin{proof} The unit vector $\frac{\surd\Delta}{\surd\Sigma}\partial_{r}$
  plays the same role as $\bb V$ in Proposition \ref{lemma}.  The proof
  follows the same argument as Proposition \ref{lemma} and is therefore
  omitted. \end{proof}
\end{proposition}
\begin{remark}
We conjecture that the vanishing of the Faraday rotation characterizes all totally geodesic submanifolds of Kerr geometry. 
\end{remark}

\section{Exact, closed form expression for the Faraday rotation in Kerr geometry}\label{ExactFaraday}
The direction 3-vector corresponding to the principal null direction $\ell$ in the symmetric frame is given in $\Sigma_{3}$ by
\begin{equation}
\label{ } 
\whector{\bb\ell}=\left[\begin{array}{c}1\\0\\0\end{array}\right],
\end{equation}
and direction 3-vector for an arbitrary photon in $\Sigma_{3}$ is given by
\begin{equation}
\label{ }
\whector{\bb K} =\frac{1}{\m P}\left[\begin{array}{c}\surd R\\ \surd\Delta\m D\\\sqrt{\Delta\Theta} \end{array}\right].
\end{equation}
Now, using (\ref{PNDgauge1})-(\ref{PNDgauge2}), we obtain the following basis for the plane of polarization:
\begin{eqnarray}
\label{basis1}
\whector{\bb b_{1}}&=&\frac1{\surd\kappa} \left[\begin{array}{c}0\\-\surd\Theta\\\m D\end{array}\right],\\
\label{basis2}
\whector{\bb b_{2}}&=& \frac1{\surd\kappa\m P} \left[\begin{array}{c}-\kappa\surd\Delta\\ \m D\surd{R}\\ \sqrt{R\Theta}\end{array}\right].  
\end{eqnarray}
We may choose the affine parameter $s$ so that $s=0$ at the event $x_{A}$ where the null geodesic intersects Alice's worldline and $s=s_{*}$ at the event $x_{B}$ where the null geodesic intersects Bob's worldline.  The basis vectors $\bb Y_{\gamma(s)}$ and $\bb Z_{\gamma(s)}$ of the plane of polarization $\mathcal P_{\gamma(s)}\in T_{\gamma(s)}M$ can now be projected onto the 2-plane of polarization in $\Sigma_{3}$, in the basis (\ref{basis1})-(\ref{basis2}):
\begin{eqnarray}
\label{parallelbasis1}
\whector{\bb y}:=\left[\begin{array}{c}\pi\left(\bb Y\right)\cdot\whector{\bb b_{1}}\\\pi\left(\bb Y\right)\cdot\whector{\bb  b_{2}}\end{array}\right]:= \frac1{\surd\Sigma}\left[\begin{array}{c}-a\cos\theta\\r\end{array}\right],\\
\label{parallelbasis2}
\whector{\bb z}=\left[\begin{array}{c}\pi\left(\bb Z\right)\cdot\whector{\bb b_{1}}\\\pi\left(\bb Z\right)\cdot\whector{\bb b_{2}}\end{array}\right]== \frac1{\surd\Sigma}\left[\begin{array}{c}-r \\ -a\cos\theta\end{array}\right].
\end{eqnarray}
Note that terms with $s$ do not survive.  All the dynamic information is
contained in the behaviour of $r$ and $\theta$.  Note as well that the polarization vector has
constant components in (\ref{parallelbasis1}) and (\ref{parallelbasis2}). 
At the event $x_{A}=\gamma(s=0)$, let Alice choose
\begin{equation}
\label{ }
\whector{\bb\digamma}_{\text{in}} = \left[\begin{array}{c}c_{1} \\c_{2}\end{array}\right] = c_{1}\whector{\bb b_{1}}+c_{2}\whector{\bb  b_{2}},
\end{equation}
that is,
\begin{equation}
\whector{\bb\digamma}_{\text{in}} = \frac1{\surd\Sigma_{0}}\left(\left(r_{0}c_{2}-c_{1}a\cos\theta_{0}\right)\whector{\bb y} - \left(c_{1}r_{0}+c_{2}a\cos\theta_{0}\right)\whector{\bb z}\right).
\end{equation}
The components of $\bb\digamma$, which stay constant with respect to the
parallel propagated frame $\bb L_{(a)}$ given by (\ref{ppframe}), are therefore
\begin{equation}
\label{ConstComponents}
\bb\digamma^{(a)}=-\frac1{\surd\Sigma_{0}}\left[\begin{array}{c}0\\0\\ c_{1}a\cos\theta_{0}-c_{2}r_{0}\\ c_{1}r_{0}+c_{2}a\cos\theta_{0}\end{array}\right].
\end{equation}
At $x_{B}=\gamma(s=s_{*})$, Bob measures $\whector{\bb\digamma}$ in the basis $\{\whector{\bb b_{1}}, \whector{\bb b_{2}}\}$, to obtain $\whector{\bb\digamma}_{\text{out}}$ which is given by (we supress the subscript for $s=s_{*}$):
\begin{eqnarray}
\whector{\bb\digamma}_{\text{out}} & = & \frac1{\surd\Sigma_{0}}\left(\left(r_{0}c_{2}-c_{1}a\cos\theta_{0}\right)\whector{\bb y}(s) - \left(c_{1}r_{0}+c_{2}a\cos\theta_{0}\right)\whector{\bb z}(s)\right)\\
\nonumber
&=&\frac1{\sqrt{\Sigma_{0}\Sigma}}\left[\begin{array}{c}c_{1}\left(r(s)r_{0}+a^{2}\cos\theta_{0}\cos\theta(s)\right)-c_{2}\left(r_{0}a\cos\theta(s)-r(s)a\cos\theta_{0}\right) \\ c_{1}\left(r_{0}a\cos\theta(s)-r(s)a\cos\theta_{0}\right)+c_{2}\left(r(s)r_{0}+a^{2}\cos\theta_{0}\cos\theta(s)\right)\end{array}\right]\\
\nonumber
&=& \frac1{\sqrt{\Sigma_{0}\Sigma}}\left[\begin{array}{cc}\left(r(s)r_{0}+a^{2}\cos\theta_{0}\cos\theta(s)\right) & -\left(r_{0}a\cos\theta(s)-r(s)a\cos\theta_{0}\right) \\\left(r_{0}a\cos\theta(s)-r(s)a\cos\theta_{0}\right) & \left(r(s)r_{0}+a^{2}\cos\theta_{0}\cos\theta(s)\right)\end{array}\right]\whector{\bb\digamma}_{\text{in}}.
\end{eqnarray}
That is, the rotation matrix in $(\ref{Faraday})$ is therefore
\begin{equation}
\frac1{\sqrt{\Sigma_{0}\Sigma}}\left[\begin{array}{cc}\left(r(s)r_{0}+a^{2}\cos\theta_{0}\cos\theta(s)\right) & -\left(r_{0}a\cos\theta(s)-r(s)a\cos\theta_{0}\right) \\\left(r_{0}a\cos\theta(s)-r(s)a\cos\theta_{0}\right) & \left(r(s)r_{0}+a^{2}\cos\theta_{0}\cos\theta(s)\right)\end{array}\right].
\end{equation}
This implies that
\begin{equation}
\label{Faraday}
\tan\chi(s)= \frac{a\left(r(s)\cos\theta_{0}-r_{0}\cos\theta(s)\right)}{r(s)r_{0}+a^{2}\cos\theta_{0}\cos\theta(s)}.
\end{equation}  

\section{Summary and discussion}
The radical simplicity of (\ref{Faraday}) stems from our geometrically
motivated choice of observers, frames, polarization plane and measurement
basis.  Exploiting the existence of the Killing-Yano tensor in Kerr
geometry, we were able to obtain a parallel propagated frame, thereby
transforming the problem of parallel transport of the polarization vector
into one of raising and lowering frame indices.  The fact that the
parallel propagated frame provides two vector fields that form a natural
basis for the plane of polarization in $T_{x}M$ at each point
$x\in\gamma(s)$ reduces the calculation of Faraday rotation to an elementary computation.

Choosing a specific class of observers in order to make it easy to compute
the result does not limit the applicability of the technique to just those
observers. Since the components in the parallel propagated frame have to stay constant, in order
to determine the Faraday rotation measured by another choice of observers,
we must apply local Lorentz transformations \emph{only at the two events}
$x_{A}$ and $x_{B}$ in order to relate the frames of the arbitrary
observers to the frames of the Carter observers going through the same
spacetime events. This is a \emph{local} transformation, quite distinct
from the geometric effect of the Kerr black hole which is a \emph{global}
phenomenon.  The analogous question in Minkowski geometry is the study of
Wigner rotation which has been extensively analyzed in the massless
case~\cite{Alsing02,Gingrich03,Terashima02}.

In order to qualitatively analyze the expression we have obtained for Faraday rotation we present some plots in Appendix B. In each of the 3 tables, the first figure (a) shows the
orbital behaviour of the null geodesic with $(r(s),\phi(s))$ as polar
coordinates, the second figure (b) depicts the same orbits in three
dimensions with spherical coordinates $(r(s),\theta(s),\phi(s))$, and the
last figure (c) depicts the Faraday rotation as a function of the affine
parameter $s$.

Tables 1 and 2 show co-rotating orbits since $\Phi>0$, while Table 3 shows
a counter-rotating orbit $(\Phi<0)$.  The apparent axial symmetry of the
Faraday rotation in Tables 1 and 3 is an artifact of our choice of initial
data, and not due to their co- and counter-rotating character.  In the
first two orbits $\dot\chi(0)<0$, while for the third one $\dot\chi(0)>0$.
The sign of $\dot\chi(0)$ is determined by the sign of the left hand side
of equation (\ref{critical}) evaluated at $s=0.$

The set of figures in Table 2 corresponds to an interesting null orbit.  A
segment of this orbit lies \emph{inside} the ergosphere (the dotted line in
figure (a)), which, somewhat surprisingly, does not seem to have a
qualitative effect on the Faraday rotation.  Photons on this orbit
circumnavigate the black hole before escaping to infinity.  That is, the
acquired azimuthal angle $\Delta\phi$ is greater than $2\pi$.  This is
possibly why $\chi$ has three critical points for this orbit.  However, this conjecture cannot be resolved without a
classification of the critical points of $\chi$, which are implicity given by
\begin{equation}
\label{critical}
\frac{r}{\surd R} + \frac{\cot\theta}{\surd\Theta}=0
\end{equation}
where $R$ and $\Theta$ are given by (\ref{R}) and (\ref{Theta}) respectively.  This is a trancendental equation with two elliptic functions with different periods.  We are unaware of any methods to obtain explicit solutions.  

Finally, we note that the measured Faraday rotation $\chi$ is invariant under the involution $L$ given by (\ref{involution}). 

The investigations of the present paper suggest a number of avenues of
further investigation.  First, and perhaps the most pressing, is to study
quantum evolution in Kerr geometry and understand how the density
matrices describing states evolve as quantum systems are exchanged between
observers.  Here one is interested in the evolution of wave packets, and not
just in the purely geometric problem of propagating polarization vectors along a null
geodesic.  Recent progress in the Cauchy problem for the Dirac equation in
Kerr geometry~\cite{Finster02,Finster03,Finster09} make it possible to
study this problem rigorously for Dirac particles.  The analogous problem
for vector particles would require new advances in our understanding of the
Maxwell equations in Kerr geometry.

In quantum information theory a central concern is coping with noise.  In
order to understand the effect of noise it would be interesting to
investigate the sensitivity of our results to perturbations of the initial
data.

Finally, we have considered the question of sharing very special frames percisely using
geometric features.  It would be very interesting to understand how to
share more general classes of frames; this is a topic that would involve
group theory as well as geometry and would make an appealing complement to the quantum information theory studies of Bartlett et
al.~\cite{Bartlett03,Bartlett07}. 
\section{Acknowledgements}
This research was supported by grants from NSERC and by a BRC grant from
the Office of Naval Research (N00001 1408 11249).  We are very grateful to
Eva Hackman, without whose help we would not be able to plot the geodesics.
Farooqui and Panangaden would like to acknowledge useful conversations with
Paul Alsing.

\bibliography{polarization}

\appendix

\section{Christoffel symbols}
The Christoffel symbols are defined by
\begin{equation}
\label{ }
\Gamma^{i}_{\phantom{i}jk}=\frac12 g^{il}(g_{lj, k}+g_{lk, j}-g_{jk, l})
\end{equation}
In Boyer-Lindquist coordinates, the nonzero ones are:
\begin{align*}
\label{}
\Gamma^{t}_{\phantom{t}r t}=& M(r^{2}+a^{2})(r^{2}-a^{2}\cos^{2}\theta)/\Sigma^{2}\Delta  \\
\Gamma^{t}_{\phantom{t}\theta t}=& -2Mra^{2}\cos\theta\sin\theta/\Sigma^{2} \\
\Gamma^{t}_{\phantom{t}r\varphi}=& aM\sin^{2}\theta(a^{4}\cos^{2}\theta-r^{2}a^{2}\cos^{2}\theta-r^{2}a^{2}-3r^{4})/\Sigma^{2}\Delta\\
\Gamma^{t}_{\phantom{t}\theta\varphi}=&2Mra^{3}\sin^{3}\theta\cos\theta/\Sigma^{2}\\
\Gamma^{r}_{\phantom{r}tt}=& M(r^{2}-a^{2}\cos^{2}\theta)\Delta/\Sigma^{3}\\
\Gamma^{r}_{\phantom{r}\varphi t}=& -aM\sin^{2}\theta(r^{2}-a^{2}\cos^{2}\theta)\Delta/\Sigma^{3}\\
\Gamma^{r}_{\phantom{r}rr}=& \left(ra^{2}\sin^{2}\theta-M(r^{2}-a^{2}\cos^{2}\theta)\right)/\Sigma\Delta\\
\Gamma^{r}_{\phantom{r}\theta r}=&-a^{2}\cos\theta\sin\theta/\Sigma \\
\Gamma^{r}_{\phantom{r}\theta\theta}=& -r\Delta/\Sigma \\
\Gamma^{r}_{\phantom{r}\varphi\varphi}=& \Delta\sin^{2}\theta\left(Ma^{2}\sin^{2}\theta(r^{2}-a^{2}\cos^{2}\theta)-r\Sigma^{2}\right)/\Sigma^{3}\\
\Gamma^{\theta}_{\phantom{\theta}tt}=& -2Mra^{2}\sin\theta\cos\theta/\Sigma^{3}\\
\Gamma^{\theta}_{\phantom{\theta}\varphi t}=& 2Mra\sin\theta\cos\theta(r^{2}+a^{2})/\Sigma^{3}\\
\Gamma^{\theta}_{\phantom{\theta}rr}=& a^{2}\sin\theta\cos\theta/\Sigma\Delta\\
\Gamma^{\theta}_{\phantom{\theta}r\theta}=& r/\Sigma\\
\Gamma^{\theta}_{\phantom{\theta}\theta\theta}=& -a^{2}\sin\theta\cos\theta/\Sigma\\
\Gamma^{\theta}_{\phantom{\theta}\varphi\varphi}=& -\cos\theta\sin\theta(\Sigma^{2}\Delta+2Mr(r^{4}+a^{2}+2r^{2}))/\Sigma^{3}\\
\Gamma^{\varphi}_{\phantom{\varphi}rt}=& Ma(r^{2}-a^{2}\cos^{2}\theta)/\Sigma^{2}\Delta\\
\Gamma^{\varphi}_{\phantom{\varphi}\theta t}=& -2Mra\cot\theta/\Sigma^{2}\\
\Gamma^{\varphi}_{\phantom{\varphi}r\varphi}=& \left((r-M)\Sigma^{2}-M(r^{2}+a^{2})(r^{2}-a^{2}\cos^{2}\theta)\right)/\Sigma^{2}\Delta\\
\Gamma^{\varphi}_{\phantom{\varphi}\theta\varphi}=& \cot\theta +2Mra^{2}\cos\theta\sin\theta/\Sigma^{2}
\end{align*}

\section{Plots}
\begin{table}[ht]
\caption{A co-rotating orbit with $\Phi=3, \kappa=12$ and initial data $r(0)=20, \theta(0)=1.57$, and  $\phi(0)=0$.}
\centering
\begin{tabular}{c}
\hline
  (a) The orbit in polar coordinates $(x=r\cos\phi,y=r\sin\phi)$.\\
\includegraphics[scale=.32]{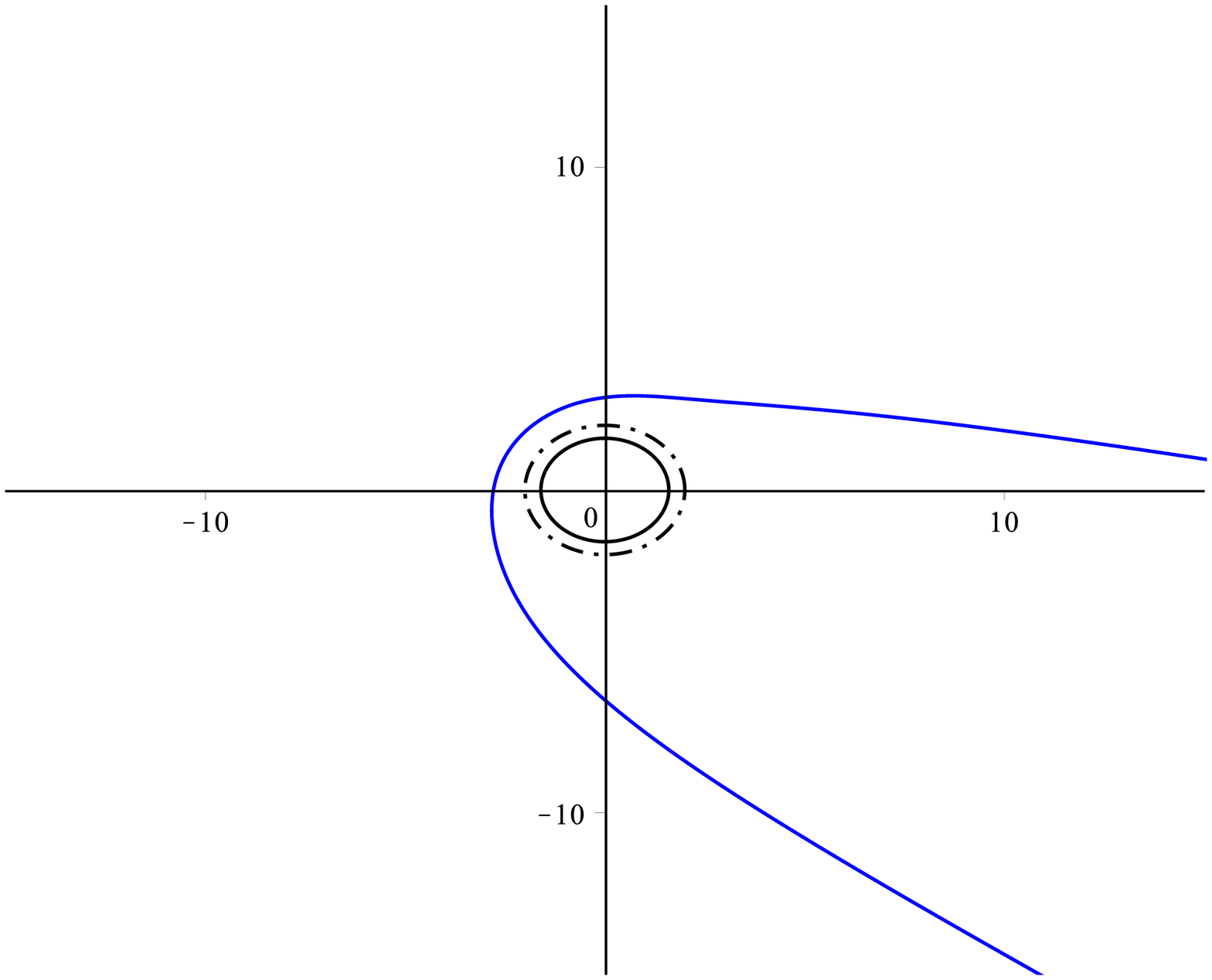}\\
\hline
  (b) The orbit in 3D spherical coordinates $\left(x=r\cos\phi\sin\theta, y=r\sin\phi\sin\theta, z=r\cos\theta\right)$.\\
\includegraphics[scale=.32]{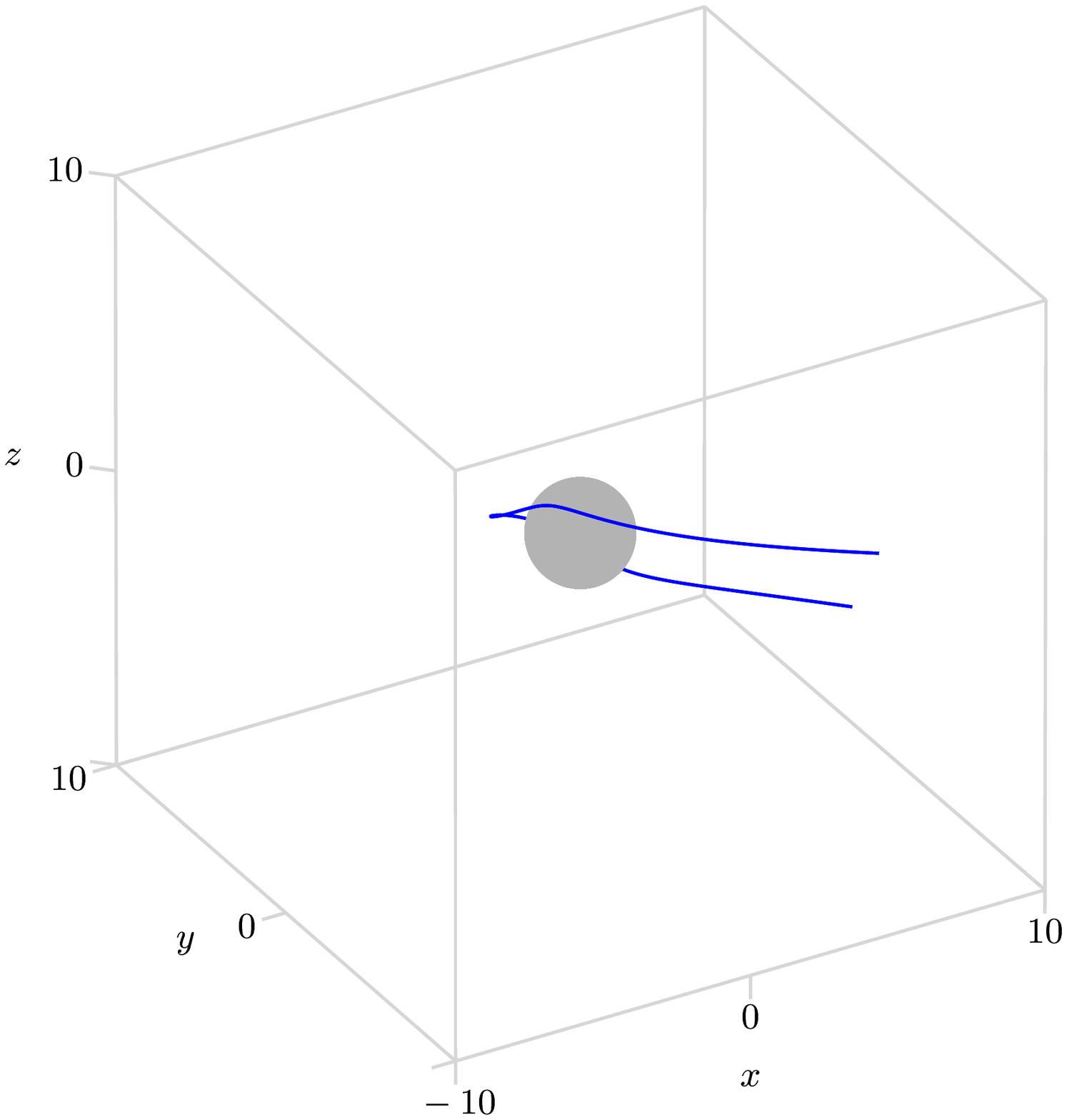}\\
\hline
  (c) The Faraday rotation angle as a function of the affine parameter $s$\\
\includegraphics[scale=.28]{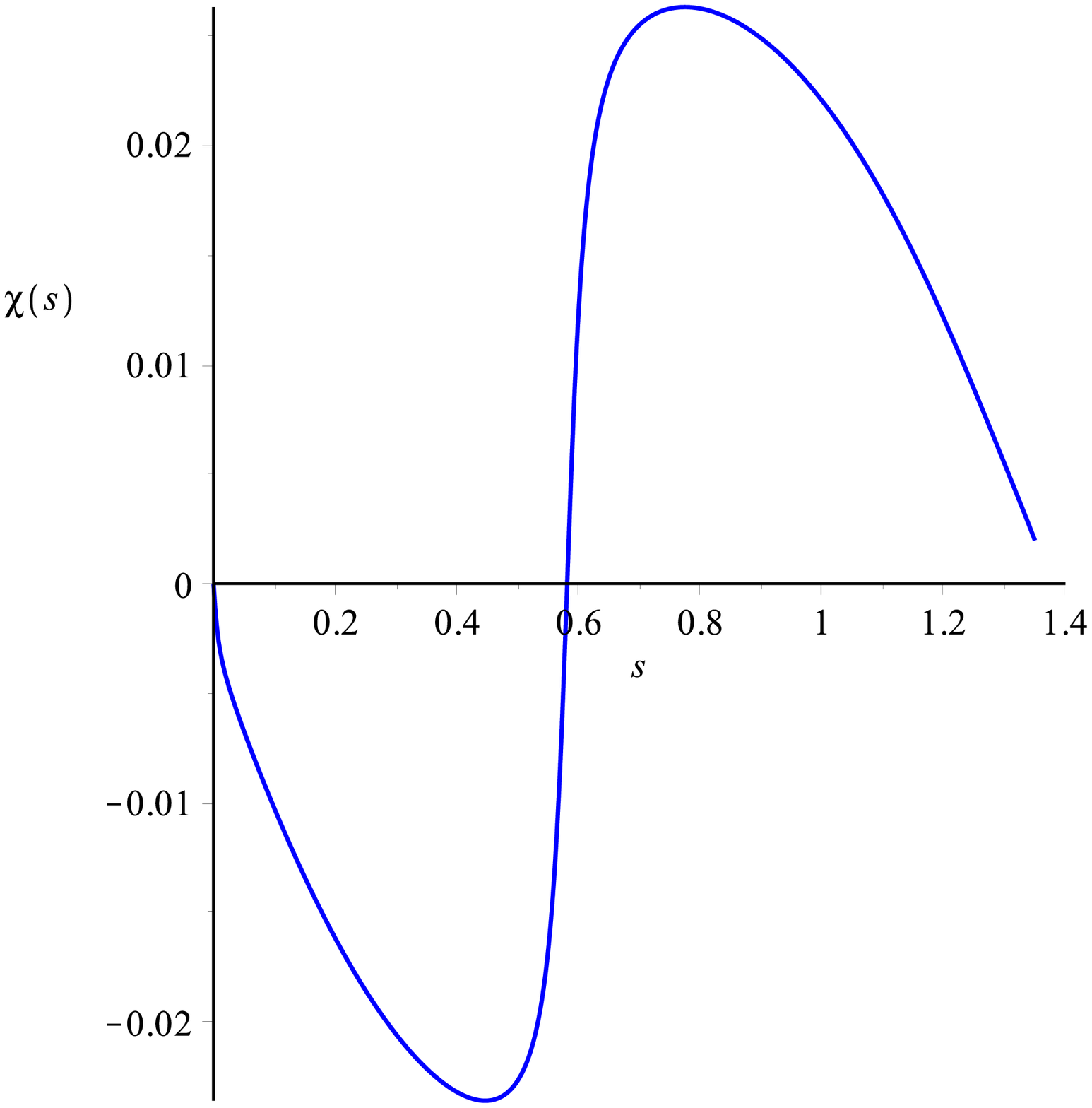}\\
\hline
\end{tabular}
\label{tab:gt}
\end{table}
\newpage
\begin{table}[ht]
\caption{A co-rotating orbit with $\Phi=3.11, \kappa=6.97$ and initial data $r(0)=20, \theta(0)=1.57$, and  $\phi(0)=0$.}
\centering
\begin{tabular}{c}
\hline
  (a) The orbit in polar coordinates $(x=r\cos\phi,y=r\sin\phi).$\\
\includegraphics[scale=.32]{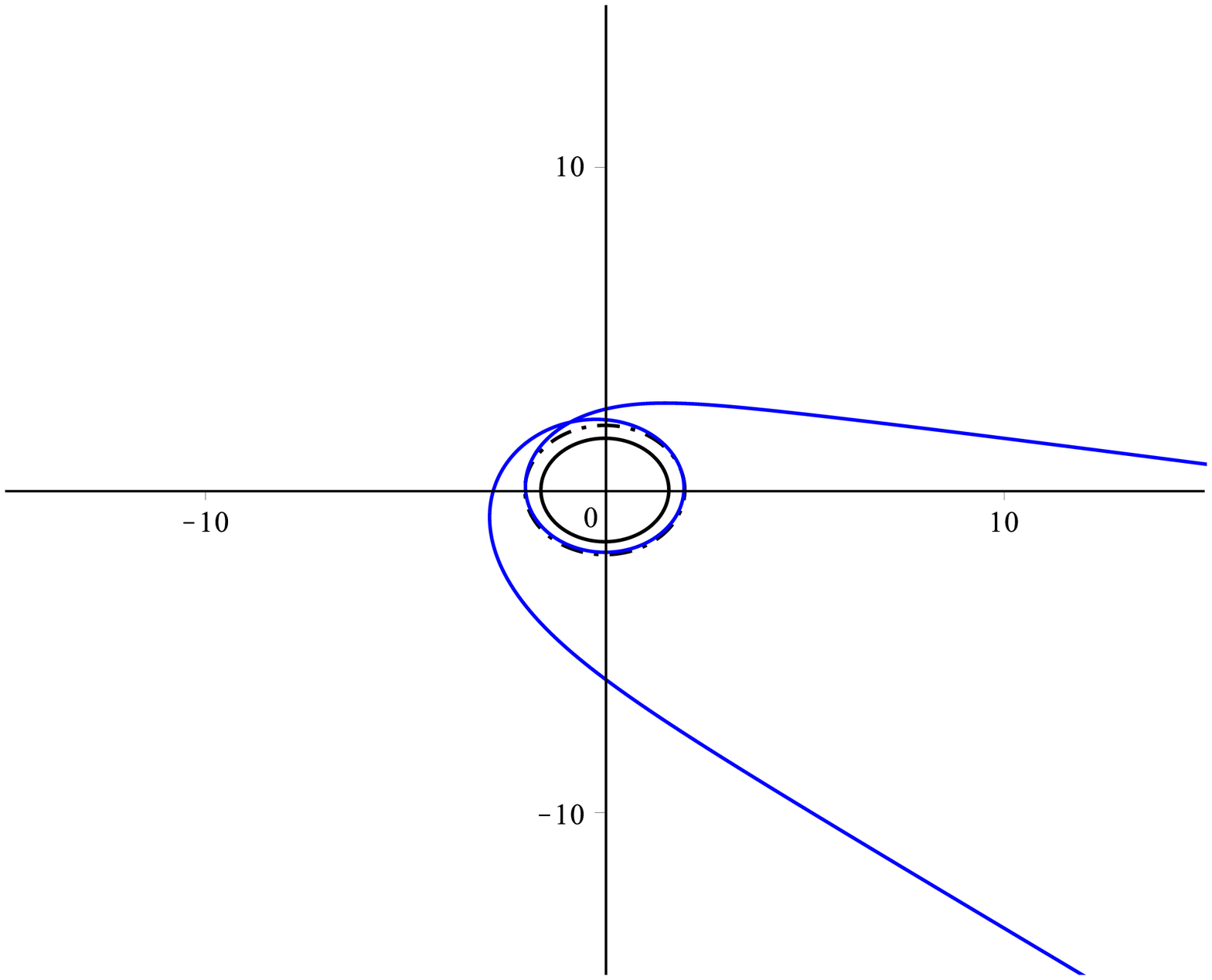}\\
\hline
  (b) The orbit in 3D spherical coordinates $\left(x=r\cos\phi\sin\theta, y=r\sin\phi\sin\theta, z=r\cos\theta\right)$.  \\
\includegraphics[scale=.32]{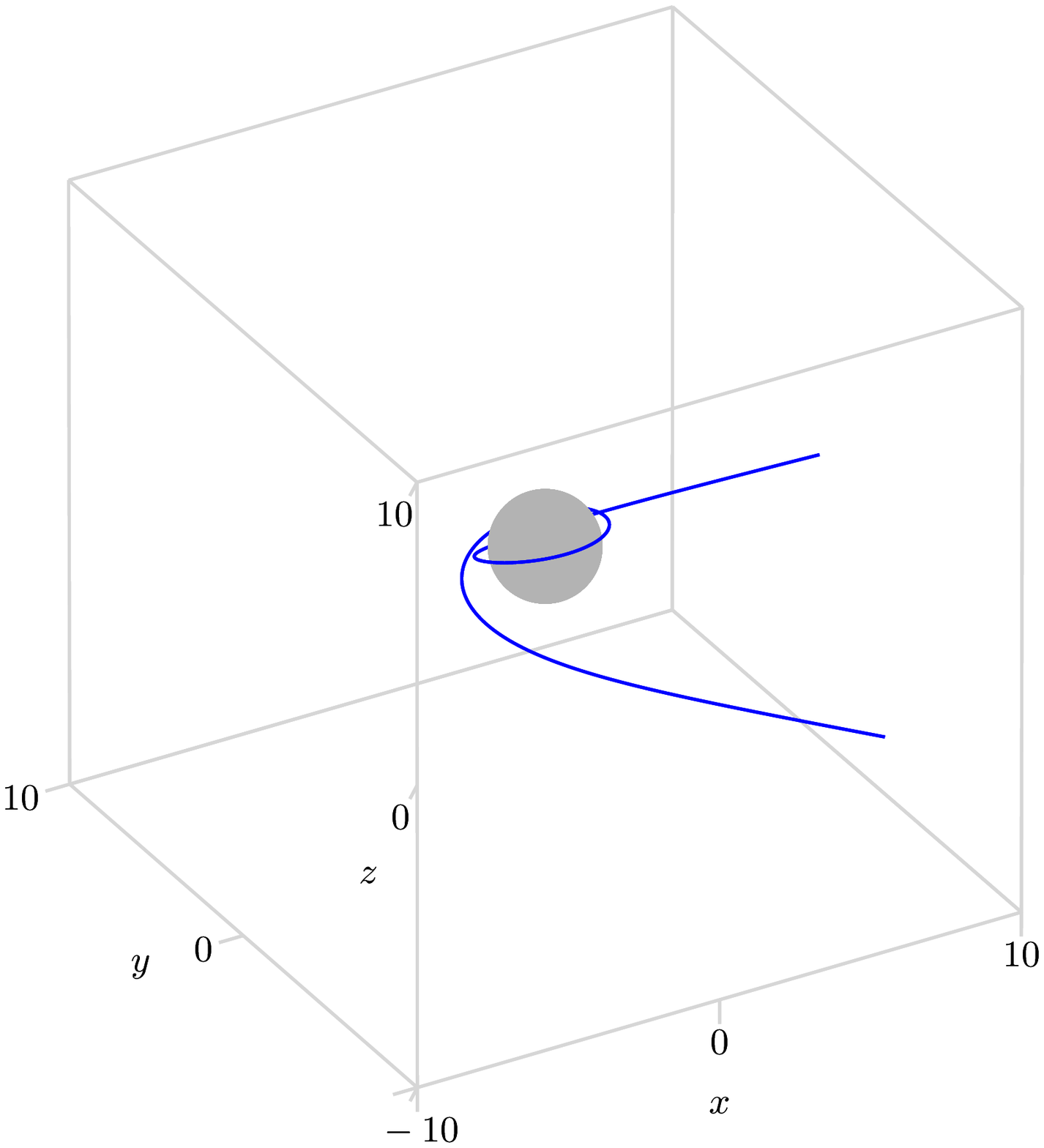}\\
\hline
  (c) The Faraday rotation angle as a function of the affine parameter $s$\\
\includegraphics[scale=.28]{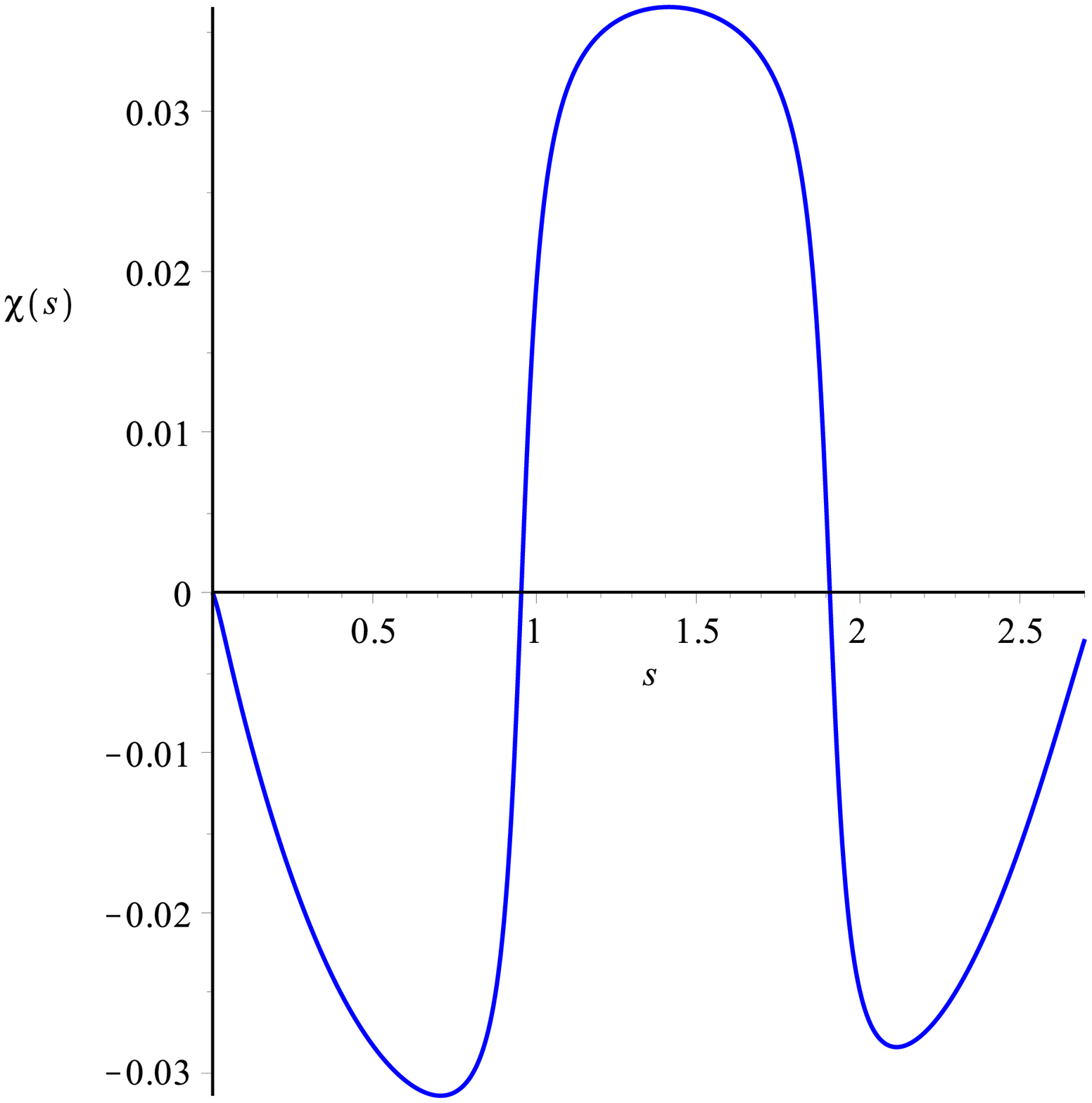}\\
\hline
\end{tabular}
\label{tab:gt}
\end{table}
\newpage
\begin{table}[ht]
\caption{A counter-rotating orbit with $\Phi=-6, \kappa=60$ and initial data $r(0)=20, \theta(0)=1.57$, and  $\phi(0)=0$.}
\centering
\begin{tabular}{c}
\hline
  (a) The orbit in polar coordinates $(x=r\cos\phi,y=r\sin\phi).$\\
\includegraphics[scale=.32]{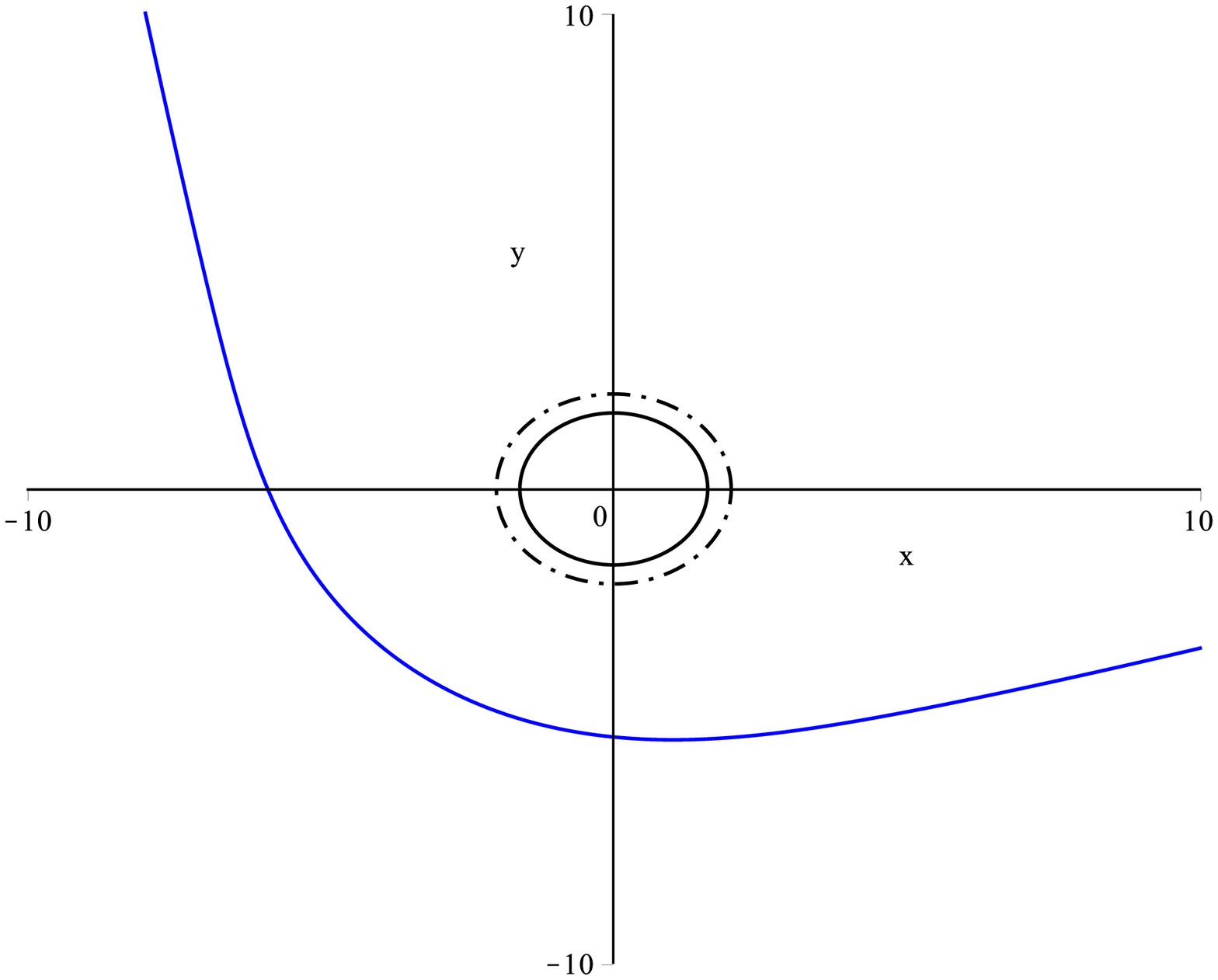}\\
\hline
  (b) The orbit in 3D spherical coordinates $\left(x=r\cos\phi\sin\theta, y=r\sin\phi\sin\theta, z=r\cos\theta\right)$.  \\
\includegraphics[scale=.32]{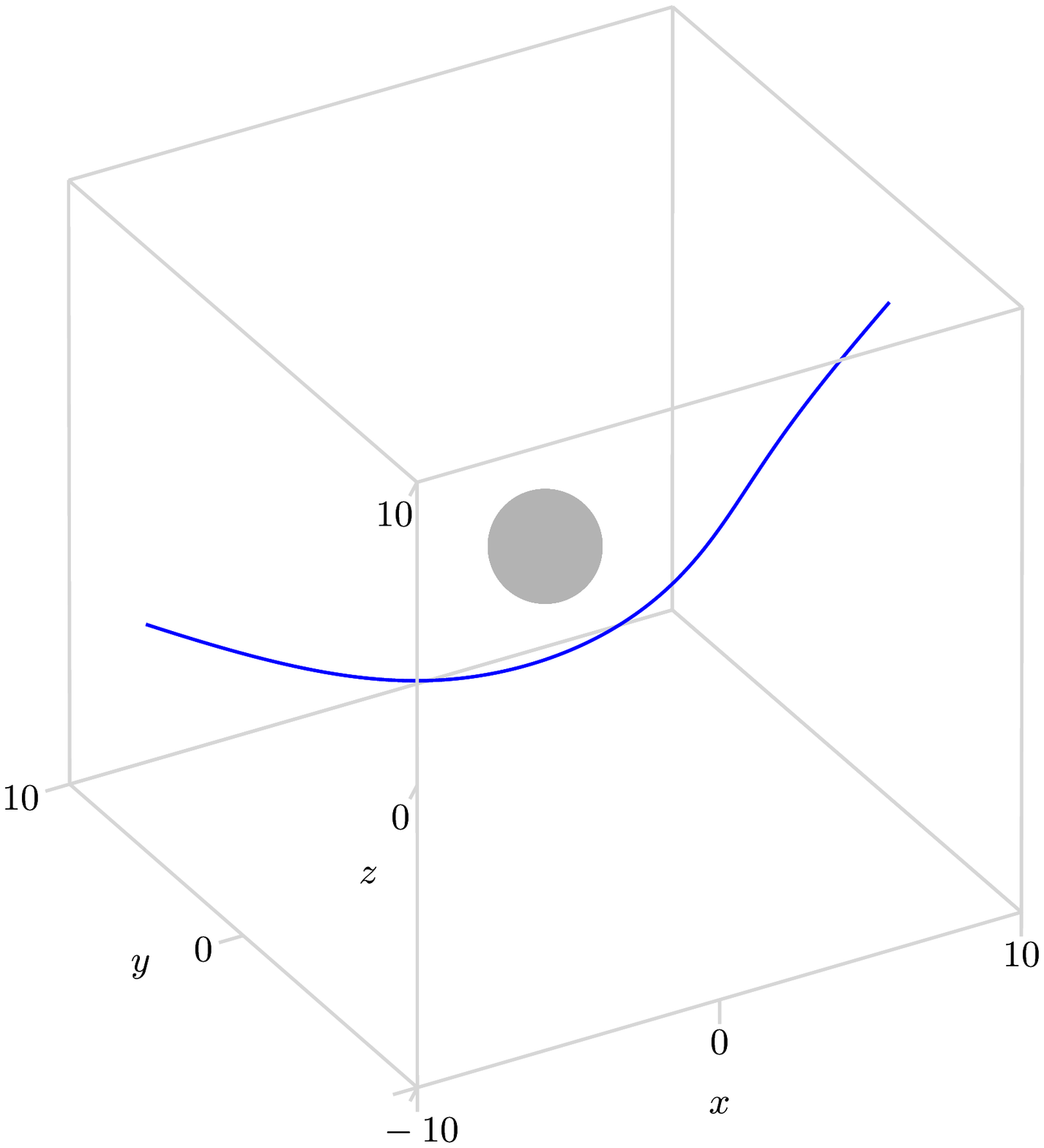}\\
\hline
  (c) The Faraday rotation angle as a function of the affine parameter $s$\\
\includegraphics[scale=.28]{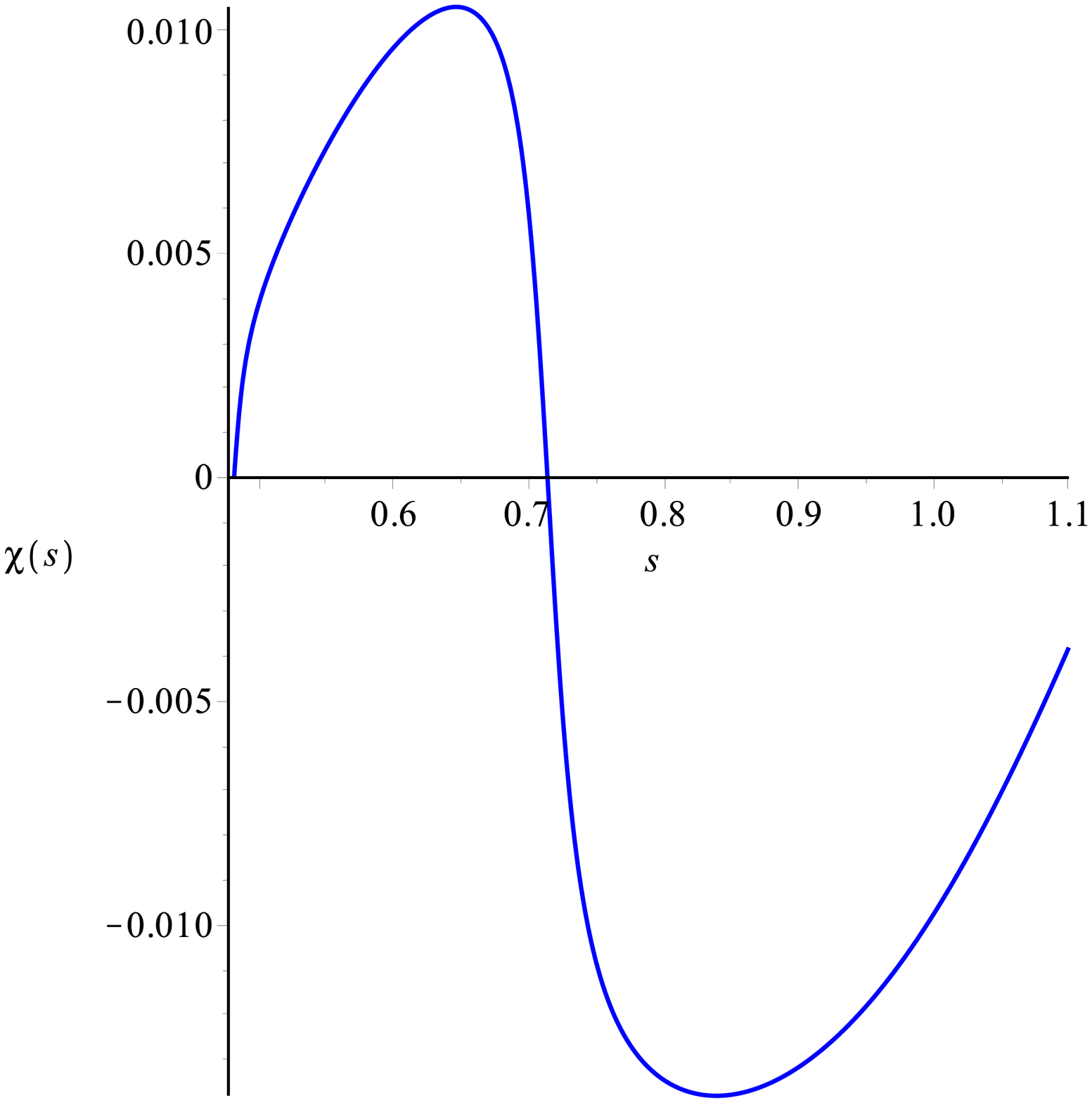}\\
\hline
\end{tabular}
\label{tab:gt}
\end{table}

\end{document}